\documentclass[letterpaper,USenglish,numberwithinsect]{lipics}

\usepackage{microtype}

\theoremstyle{plain}
\newtheorem{proposition}[theorem]{Proposition}
\theoremstyle{remark}
\newtheorem{claim}[theorem]{Claim}

\newtheorem{remark}[theorem]{Remark}

\newcommand{\BPP}{\mathrm{BPP}}
\newcommand{\ZPP}{\mathrm{ZPP}}
\newcommand{\NP}{\mathrm{NP}}
\newcommand{\PNP}{\PTIME^\NP}
\newcommand{\Ppoly}{\mathrm{P/poly}}
\newcommand{\PSPACE}{\mathrm{PSPACE}}
\newcommand{\N}{\mathbb{N}}
\newcommand{\co}{\mathrm{co}}
\newcommand{\io}[1]{\mathrm{i.o}\mathchar`-#1}
\newcommand{\PTIME}{\mathrm{P}}
\newcommand{\EXP}{\mathrm{EXP}}
\newcommand{\EXPNP}{\EXP^\NP}
\newcommand{\MIP}{\mathrm{MIP}}
\newcommand{\NEXP}{\mathrm{NEXP}}
\newcommand{\BPPlog}{\BPP \doubleslash \log}

\newcommand{\PH}[1]{\mathrm{\Sigma}_{#1}^{\mathrm{p}}}
\newcommand{\coPH}[1]{\mathrm{\Pi}_{#1}^{\mathrm{p}}}
\newcommand{\SH}[1]{\mathrm{S}_{#1}^{\mathrm{p}}}
\newcommand{\SHEXP}{\mathrm{S}_{2}^{\mathrm{exp}}}

\newcommand{\reduction}[2][p]{\le^{#1}_{#2}}

\def\binset{\{0, 1\}}
\def\binstr{\binset^*}
\newcommand{\NumSet}[1]{\{1, \cdots, #1\}}
\newcommand{\ESet}[1]{\{\,#1\,\}}
\newcommand{\ISet}[2]{\{\,#1 \mid #2 \,\}} 
\newcommand{\supplement}[1]{\ \left(\,#1\,\right)}

\def\doubleslash{/\mkern-4mu/}

\def\eg{{\it e.g.\ }}
\newcommand{\etal}[1]{#1 {\it et al.}}
\def\ie{{\it i.e.\ }}

\let\subset\undefined
\let\subset\subseteq

\let\phi\undefined
\let\phi\varphi

\newcommand\restr[2]{#1|_{#2}}

\def\pa#1{v_#1}
\def\stage#1{^{(#1)}}

\def\binput{b_\mathrm{in}}
\def\lexass{{V_\Phi}}
\newcommand{\satinstance}[2]{F_{#1}(#2)}
\newcommand{\extass}{{\hat f_1}}

\title{Identifying an Honest $\EXPNP$ Oracle Among Many}

\author{Shuichi Hirahara}
\affil{Department of Computer Science, The University of Tokyo\\
  7-3-1 Hongo, Bunkyo-ku, Tokyo 133-8654, Japan\\
  \texttt{hirahara@is.s.u-tokyo.ac.jp}}
\authorrunning{S. Hirahara} 
\Copyright{Shuichi Hirahara} 
\subjclass{F.1.1 Models of Computation; F.1.2 Modes of Computation; F.1.3 Complexity Measures and Classes}

\keywords{nonuniform complexity, short advice, instance checker, interactive proof systems, probabilistic checkable proofs}

\volumeinfo{David Zuckerman}
{1}
{30th Conference on Computational Complexity (CCC'15)}
{33}
{1}
{1}
\EventShortName{CCC 2015}
\DOI{10.4230/LIPIcs.CCC.2015.p}
\serieslogo{}

\begin{document}

\maketitle

\begin{abstract}
  We provide a general framework to remove short advice
  by formulating the following computational task for a function $f$:
  given two oracles at least one of which is honest (\ie correctly computes $f$ on all inputs)
  as well as an input,
  the task is to compute $f$ on the input with the help of the oracles by a probabilistic polynomial-time machine,
  which we shall call a \emph{selector}.
  We characterize the languages for which short advice can be removed by the notion of selector:
  a paddable language has a selector if and only if
  short advice of a probabilistic machine that accepts the language can be removed under any relativized world.

  Previously, instance checkers have served as a useful tool to remove short advice of probabilistic computation.
  We indicate that existence of instance checkers is a property stronger than that of removing short advice:
  although no instance checker for $\EXPNP$-complete languages exists unless $\EXPNP = \NEXP$,
  we prove that there exists a selector for any $\EXPNP$-complete language, by building on the proof of $\MIP = \NEXP$ by Babai, Fortnow, and Lund (1991).
\end{abstract}

\section{Introduction}
Blum and Kannan~\cite{BK95} introduced the notion of \emph{instance checker}.
Roughly speaking,
an instance checker for a function $f$ is an efficient probabilistic machine that,
given access to an oracle,
checks if the oracle computes $f(x)$ correctly on a given instance $x$;
the oracle models a possibly buggy program that purports to compute $f$,
and an instance checker verifies whether the program works correctly on a given instance.

The notion of instance checker is intimately related
to interactive proof systems:
the line of work showing the power of interactive proofs~\cite{LFKN92,Sha92,BFL91}
yielded
instance checkers for $\PTIME^{\#\PTIME}$-, $\PSPACE$-, and $\EXP$-complete languages;
in addition,
Blum and Kannan~\cite{BK95} gave a characterization of the languages with an instance checker 
by a function-restricted interactive proof system.
Since any language with an interactive proof protocol is in $\NEXP$~\cite{FRS94},
any language with an instance checker must be in $\NEXP \cap \co\NEXP$.

In this paper,
we investigate a computational task weaker than instance checking of a (Boolean) function $f$:
we are given access to two oracles (instead of a single oracle) as well as an input $x$;
again, both of the oracles purport to compute $f$;
however,
it is assumed that at least one of the two oracles is \emph{honest}, \ie computes $f(q)$ correctly on all inputs $q$;
and the task is to compute $f(x)$ with the help of the oracles in polynomial time.
We shall call a probabilistic machine doing the task a \emph{(probabilistic) selector} for $f$.

If the answers of oracles on the input $x$ agree, then we have only to output the answer,
which is surely correct by the assumption.
Thus, the task of a selector is essentially to identify the honest oracle
when two oracles disagree on $x$ (\ie
one of the oracles asserts that $f(x) = 0$, whereas the other asserts that $f(x) = 1$).

Our main result shows that there exists a selector for $\EXPNP$-complete languages.
We also show that the notion of selector does not change even if there are one honest oracle and polynomially many dishonest oracles.
Thus, these results can be encapsulated in the following phrase:
``identifying an honest oracle among many is \emph{strictly} weaker than instance checking unless $\EXPNP = \NEXP$.''

Although the task is weaker than instance checking,
a situation in which one may assume existence of an honest oracle naturally arises
out of computation with advice:
Suppose, for example, that a (paddable) language $L$ is computed by a probabilistic machine $M$ with advice of one bit.
We regard $M$ with advice $0$ and $1$ as two oracles $A_0$ and $A_1$, respectively.
By the definition of advice, either $A_0$ or $A_1$ is honest on all the inputs (of the same length).
Thus, the advice of one bit can be removed if $L$ has a selector.
We can in fact remove advice of size $O(\log n)$,
since a selector can identify an honest oracle among polynomially many oracles.

\subsection{Removing Short Advice for Probabilistic Computation}

In early work as to removing short advice for probabilistic computation,
Trevisan and Vadhan~\cite{TV07}
gave an insight into the potential of instance checkability:
they demonstrated that
instance checkability can be exploited to
remove short advice.
Based on the existence of an instance checker for $\EXP$-complete languages,
they showed 
a quantitative tradeoff from a uniform worst-case-hardness assumption (\ie $\EXP \not\subset {\rm BPTIME}(t(n^{O(1)}))$) to
average-case hardness of $\EXP$ (\ie $\EXP$ contains languages that cannot be solved by probabilistic computation on a fraction better than $\frac12 + \frac{1}{t}$ of inputs in time $t$).

They also argued that
their result cannot be obtained via \emph{black-box uniform reductions}.
Typical constructions of a worst-case to average-case connection are
based on the following scheme:
we convert a function $f$ into another function $f'$, which is an error-correcting code of $f$;
if we have a ``black-box'' algorithm that computes $f'$ on a fraction greater than $\frac12 + \epsilon$ of the inputs,
then a probabilistic machine that takes advice can compute $f$ on all inputs by decoding $f'$.
Since it is impossible to uniquely decode $f'$ for small $\epsilon$,
the advice is used to identify $f$ and is provably indispensable.

Indeed, it was the instance checkability of $\EXP$-complete languages
that broke the black-box construction in the proof of Trevisan and Vadhan;
the instance checkability enabled them to remove advice of logarithmic size.
Therefore,
it will be helpful for future research to closely understand the property that they actually exploited.

Subsequent to their work,
instance checkability has since been exploited to cope with short advice for probabilistic computation:
for example,
Barak~\cite{Bar02} proved the first hierarchy theorem for probabilistic computation with short advice;
Buhrman, Fortnow, and Santhanam~\cite{BFS09}
unconditionally separated 
${\rm BPEXP}$ from $\BPP$ with advice of subpolynomial size;
and
Buhrman, Fortnow, Kouck\'{y}, and Loff~\cite{BFKL10} gave some evidences that
a deterministic efficient computation with oracle access to the set of Kolmogorov-random strings
can be simulated by
a probabilistic efficient computation.

\subsection{Our Results}

In fact, the notion of selector captures a property of removing short advice:
\begin{theorem}
  \label{thm:characterizationelim}
  Let $L$ be an arbitrary paddable language.
  The following are equivalent:
  \begin{enumerate}
    \item
      There exists a selector for $L$.
    \item
      For any oracle $R \subset \binstr$,
      it holds that
      $L \in \BPP^R \doubleslash \log$ implies $L \in \BPP^R$.
  \end{enumerate}
\end{theorem}
That is, a paddable language has a selector if and only if 
short advice can be removed under any relativized world.
($``\doubleslash''$ means advice that can depend on coin flips of probabilistic machines as well as input length~\cite{TV07}.)

In addition, we construct a selector for $\EXPNP$-complete languages,
thereby indicating an essential difference between selectors and instance checkers.
We also give an upper bound on the languages with a selector:
\begin{theorem}
  [Main Theorem]
  \label{thm:honesty} \hspace{0em}
  \begin{enumerate}
    \item
      \label{enum:honestylower}
      Every $\EXPNP$-complete language has a selector.
    \item
      \label{enum:honestyupper}
      Any language with a selector is in $\SHEXP$
      (which is an exponential-time analogue of $\SH2$).
  \end{enumerate}
\end{theorem}

Thus, existence of an instance checker is a property stronger than that of removing short advice (or, equivalently, existence of a selector):
although no instance checker for $\EXPNP$-complete languages exists unless $\EXPNP = \NEXP$,
short advice of a probabilistic machine that accepts $\EXPNP$-complete languages can be removed.

\subsubsection*{Our Techniques}
The most technical part of this paper is a proof of the main theorem (Theorem~\ref{thm:honesty}, Part~\ref{enum:honestylower}).
In order to construct a selector for $\EXPNP$-complete languages,
we build on the proof of $\MIP = \NEXP$ by Babai, Fortnow, and Lund~\cite{BFL91}.
As pointed out by G\'abor Tardos in the paper~\cite{BFL91},
the complexity of honest provers of the interactive proof system for $\NEXP$-complete languages can be bounded above by $\EXPNP$.
We crucially use this fact to check satisfiability of an exponential-sized formula with the help of an $\EXPNP$-complete oracle.
We also compare two exponential-sized strings by performing a binary search.

Thanks to plenty of machinery that has been cultivated together with interactive proof systems, program checking, and PCPs,
we can prove the main theorem by careful combinations of such machinery.
For example,
we exploit a multilinearity test~\cite{BFL91} and the self-correction of low-degree polynomials~\cite{BF90,Lip91}.

Due to the usage of arithmetization, we suspect that our proof of the main theorem does algebrize~\cite{AW09} but does not relativize.

\subsubsection*{Variants of Selectors}

We also investigate other variants of selectors:
a \emph{deterministic selector} and a \emph{nonadaptive deterministic selector}.
We focus on the ``suprema'' of the languages with a selector, namely, upper bounds on these languages
and existence of a selector for languages complete for a complexity class that is close to the upper bounds.
(Note that the languages with a selector are not necessarily closed downward.
For example, although $\NEXP \subset \EXPNP$,
we do not know whether $\NEXP$-complete languages have a selector or not.)

For a nonadaptive deterministic selector,
we prove polynomial-time analogues of Theorem~\ref{thm:honesty}:
\begin{theorem}
  \label{thm:honestytt} \hspace{0em}
  \begin{enumerate}
    \item
      \label{enum:honestyttlower}
      Every $\PNP$-complete language has a nonadaptive deterministic selector.
    \item
      \label{enum:honestyttupper}
      Any language with a nonadaptive deterministic selector is in $\SH2$.
  \end{enumerate}
\end{theorem}
The proofs of this theorem will clearly illustrate the basic ideas for Theorem~\ref{thm:honesty}.

Notice that
$\PNP$ is close to the upper bound $\SH2$
since $\PNP \subset \SH2 \subset \ZPP^\NP$ \cite{RS98,Cai07}.
(Under suitable hardness assumptions, it holds that $\PNP = \SH2$ by derandomization \cite{KM02}.)

For a deterministic selector, the supremum is $\PSPACE$:
\begin{theorem}
  \label{thm:honestyT} \hspace{0em}
  \begin{enumerate}
    \item
      \label{enum:honestyTlower}
      Every $\PSPACE$-complete language has a deterministic selector.
      More generally, any downward self-reducible language has a deterministic selector.
    \item
      \label{enum:honestyTupper}
      Any language with a deterministic selector is in $\PSPACE$.
  \end{enumerate}
\end{theorem}

As with Theorem~\ref{thm:characterizationelim}, a property of removing short advice for deterministic computation
can be characterized by existence of a deterministic selector:
\begin{theorem}
  \label{thm:characterizationdet}
  Let $L$ be an arbitrary paddable language.
  The following are equivalent:
  \begin{enumerate}
    \item
      There exists a deterministic selector for $L$.
    \item
      For any oracle $R \subset \binstr$,
      it holds that
      $L \in \PTIME^R / \log$ implies $L \in \PTIME^R$.
  \end{enumerate}
\end{theorem}

\subsection{Comparison with Prior Work}

In seminal work by Karp and Lipton~\cite{KL82} as to collapses of a uniform class contained in a nonuniform class,
it was shown that $\NP \subset \PTIME/\log$ implies $\NP \subset \PTIME$ and $\PSPACE \subset \PTIME/\log$ implies $\PSPACE \subset \PTIME$.
These results are essentially equivalent to the existence of deterministic selectors for $\NP$- and $\PSPACE$-complete languages, respectively.

Fortnow and Klivans~\cite{FK05} observed that $\NEXP \subset \BPPlog$ implies $\NEXP = \BPP$ by combining previous results.
Similarly, it is folklore that $\EXPNP \subset \BPPlog$ implies $\EXPNP = \BPP$.
This follows by combining the result by Buhrman and Homer~\cite{BH92} stating that $\EXPNP \subset \EXP/{\rm poly}$ implies $\EXPNP = \EXP$,
the existence of an instance checker (or a selector) for $\EXP$-complete languages, and $\BPPlog \subset \Ppoly$ (see \cite{FK05}).

We clarify the differences between the folklore and our results in two respects.
First, our results can be relativized on the right-hand side.
Second, selectors can be used to quantitatively remove advice of logarithmic size:
if we allow a machine to run in time $t$ (instead of polynomial time), then advice of size $\log t$ can be removed.
\begin{corollary}
  [Analogous to Proposition 5.6 in \cite{TV07}]
  \label{cor:quantitative}
  There are an $\EXPNP$-complete language $L$ and a constant $d \in \N$ such that,
  for any nice time bound\footnote{
    Although the definition of a \emph{nice time bound} is the same as in \cite{TV07},
    we note that the condition $t(n) \le 2^n$ is not needed here.
  } $t \colon \N \to \N$ and any oracle $R \subset \binstr$,
  if $L \in {\rm BPTIME}^R(t(n)) \doubleslash \log t(n)$ then $L \in {\rm BPTIME}^R(t(n^d))$.
\end{corollary}

We mention in passing that, by substituting selectors for instance checkers in the proofs of Trevisan and Vadhan~\cite{TV07},
one can obtain a quantitative tradeoff from a uniform worst-case-hardness assumption on $\EXPNP$ to
a uniform average-case hardness of $\EXPNP$ (see \cite[Theorem 5.7]{TV07}).

\subsection{Application: Random Strings vs.\ Randomized Computation}

In Section \ref{sec:app},
we will give another application in order to demonstrate usefulness of the notion of selector,
by simply substituting selectors for instance checkers in the previous work
by Buhrman, Fortnow, Kouck\'{y}, and Loff~\cite{BFKL10}.

They tried to show that
a deterministic polynomial-time computation with oracle access to the set of Kolmogorov-random strings
is, in some sense, equivalent to
a probabilistic polynomial-time computation;
they modeled oracle access to the set of Kolmogorov-random strings as advice strings of high nonuniform complexity.
Although the nonuniform complexity of the advice strings is required to be much higher than that of Kolmogorov-random strings,
they showed, as a partial result, that if a language $L$ can be solved in deterministic polynomial time with high nonuniform advice,
then $L$ is in $\BPP$ with advice of almost linear size~\cite[Theorem 13]{BFKL10}.

Because the goal is to show that $L$ is in $\BPP$ \emph{without} any advice,
they further observed that
one can dispense with the advice of almost linear size
if there exists an instance checker for $L$.
From this observation,
they showed that,
for any class $\mathcal C \in \{\NP,\PTIME^{\#\PTIME},\PSPACE,\EXP\}$,
if some $\mathcal C$-complete language can be solved in deterministic polynomial time with high nonuniform advice,
then
$\mathcal C \subset \BPP$~\cite[Theorem 15]{BFKL10}.

In fact, they proved this result by analyzing the two cases:
For $\mathcal C \in \{\PTIME^{\#\PTIME}, \allowbreak \PSPACE, \allowbreak \EXP\}$,
they used an instance checker for $\mathcal C$-complete languages, whose existence was shown by \cite{LFKN92,Sha92,BFL91};
Unfortunately,
because
it is not known whether $\NP$-complete languages have instance checkers or not,
they needed to prove the result in another way solely for $\mathcal C = \NP$.

The notion of selector, however, enables us to show the result in a unified way
and to extend the result from
$\{\NP,\allowbreak\PTIME^{\#\PTIME},\allowbreak\PSPACE,\allowbreak\EXP\}$ to any classes whose complete languages have a selector.
Given the fact that many languages have selectors (\eg languages with instance checkers and downward self-reducible languages),
it becomes more plausible that we can dispense with the advice of almost linear size;
thereby we slightly strengthen the connection between Kolmogorov-random strings and randomized computation.

\subsubsection*{Organization}

In Section \ref{sec:basics},
we give
formal definitions, common properties of selectors, and a proof of Theorem \ref{thm:characterizationelim}.
Sections \ref{sec:truth}, \ref{sec:bpp}, and \ref{sec:turing} are 
devoted to investigating
nonadaptive deterministic selectors,
probabilistic selectors, and
deterministic selectors,
respectively.
We mention some possible directions for future work in Section \ref{sec:conclusions}.

\subsection*{Preliminaries and Notations}

We assume that the reader is familiar with basics of computational complexity (\eg \cite{AB09}).

For a Turing machine $M$, let $M(x)$ denote the output of $M$ on input $x \in \binstr$.
For an oracle Turing machine $M$ and oracles $A_0, A_1 \subset \binstr$,
let $M^{A_0, A_1}$ represent a machine equipped with access to oracle $A \subset \binstr$ such that
$A(i \cdot q) = A_i(q)$, for each $i \in \binset$ and for any $q \in \binstr$.
We identify false and true with $0$ and $1$, respectively.
We also identify a language $L \subset \binstr$ with its characteristic function from $\binstr$ to $\binset$.
For a Boolean formula $\phi$ in $n$ variables, we abuse notation and write $\phi \colon \binset^n \to \binset$.

We say that a language $L$ is paddable if there exists a polynomial-time machine that,
on input $(x, 1^m)$ where $x \in \binset^n$ and $n \le m$,
outputs a string $y$ of length $m$ such that
$y \in L$ if and only if $x \in L$.

\section{Definitions and Common Properties of Selectors}
\label{sec:basics}
In this section,
we give formal definitions of selectors
and show common properties that all types of selectors have.
First, we define a probabilistic selector:
\begin{definition}[Probabilistic Selector]
  \label{def:honesty}
  A \emph{(probabilistic) selector} $S$ for a language $L \subset \binstr$
  is
  a probabilistic polynomial-time oracle Turing machine
  which
  computes $L$ with high probability,
  given arbitrary two oracles $A_0, A_1 \subset \binstr$ such that $A_0$ or $A_1$ is equal to $L$.
  That is, for any input $x \in \binstr$ and oracles $A_0, A_1 \subset \binstr$,
  \begin{align*}
    L \in \ESet{ A_0, A_1 }
    \implies
    \Pr \left[ S^{A_0, A_1} (x) = L(x) \right] \ge \frac{2}{3},
  \end{align*}
  where the probability is taken over coin flips of $S$.
\end{definition}
Note that the success probability $\frac 2 3$ in Definition \ref{def:honesty} can be enhanced by repetitions.
We often abbreviate a probabilistic selector as a selector.

An oracle equal to $L$ is said to be \emph{honest}; otherwise it is said to be \emph{dishonest}.

Next, we define a deterministic selector and a nonadaptive deterministic selector:
\begin{definition}[Deterministic Selector]
  \label{def:honestyT}
  A \emph{deterministic selector} for a language $L$ is
  a deterministic polynomial-time oracle machine $S$
  such that
  $S^{L, X} (x) = S^{X, L} (x) = L(x)$
  for any oracle $X \subset \binstr$ and for any input $x \in \binstr$.
\end{definition}
\begin{definition}[Nonadaptive Deterministic Selector]
  \label{def:honestytt}
  A \emph{nonadaptive deterministic selector} $S$ for a language $L$ is
  a deterministic polynomial-time oracle machine
  such that
  \begin{itemize}
    \item
      $S^{L, X} (x) = S^{X, L} (x) = L(x)$ for any oracle $X \subset \binstr$ and any input $x \in \binstr$, and
    \item
      $S$ is nonadaptive, \ie
      there exists a polynomial-time machine which, on input $x \in \binstr$, outputs the query set $Q(x)$
      of all the queries that $S$ makes to either of the oracles.
  \end{itemize}
\end{definition}

We state a useful structural property:
\begin{proposition}
  \label{prop:closedequiv}
  The class of the languages with a selector is closed under polynomial-time Turing equivalence.
  Namely,
  $L_1 \reduction T L_2$ and $L_2 \reduction T L_1$ imply that if $L_1$ has a selector then so does $L_2$.

  In particular,
  it is closed under complement.
  Moreover,
  for any complexity class $\mathcal C$, if a \emph{specific} $\mathcal C$-complete language
  has a selector, then so does an \emph{arbitrary} $\mathcal C$-complete language.
\end{proposition}
\begin{proof}
  The proof is essentially the same with Beigel's theorem~\cite{BK95}, which shows the same closure property of instance checkers.
  The idea is as follows: reduce a $L_2$ problem to a $L_1$ problem by using the reducibility from $L_2$ to $L_1$,
  and solve the $L_1$ problem by running a selector for $L_1$, while converting its query (which is an instance of $L_1$) into an instance of $L_2$.

  Let $M_{ij}$ be a polynomial-time oracle machine that witnesses the polynomial-time Turing reduction $L_i \reduction T L_j$ for each $(i, j) \in \ESet{(1, 2), (2, 1)}$
  (that is, $M_{ij}^{L_j}(x) = L_i(x)$ for any $x$),
  and $S$ be a selector for $L_1$.
  The following algorithm yields a selector for $L_2$:
  Given an input $x \in \binset^n$ and two oracles $A_0, A_1$, simulate $M_{21}(x)$ in order to compute $L_2(x)$.
  If $M_{21}$ makes a query $q$,
  then we try to answer it with $L_1(q)$, by running $S(q)$.
  If $S$ makes a query $q'$ to the $i$th oracle $(\,i\in\binset\,)$,
  then answer it with $M^{A_i}_{12}(q')$.

  Let $A_i$ be an honest oracle (\ie $A_i = L_2$).
  Then, we have $M^{A_i}_{12}(q') = M^{L_2}_{12}(q') = L_1(q')$,
  and hence $S(q)$ is simulated under the existence of the honest oracle; thus it outputs $L_1(q)$ correctly with high probability (say, with probability at least $1 - 2^{-n}$, by running the selector $O(n)$ times).
  Therefore, the simulation of $M_{21}(x)$ results in outputting $L_2(x)$ with probability at least $1 - 2^{-n} n^{O(1)}$.
\end{proof}

\begin{remark}
  \label{remark:closedequiv}
  Similarly, the class of languages with a deterministic selector is closed under polynomial-time Turing equivalence,
  and the class of languages with a nonadaptive deterministic selector is closed under polynomial-time truth-table (\ie nonadaptive) equivalence.
\end{remark}

To prove Theorem \ref{thm:characterizationelim},
we show that
the definitions of selectors are robust even if
we consider a situation in which we are given polynomially many oracles.
\begin{lemma}
  \label{lemma:robustmany}
  For any language $L \subset \binstr$, the following are equivalent:
  \begin{enumerate}
    \item
      There exists a selector for $L$.
    \item
      There exists a selector for $L$ that identifies an honest oracle among polynomially many oracles.
  \end{enumerate}

The latter can be formally stated as follows:
  for any polynomial $m \colon \N \to \N$,
  there exists a probabilistic polynomial-time oracle Turing machine $S$ such that,
  on input length $n \in \N$,
  it holds that
  $ \Pr \left[ S^A (x) = L(x) \right] \ge \frac 2 3 $
  for any $x \in \binset^n$,
  where $A$ is an arbitrary oracle such that
  there exists an index $i \in \NumSet{m(n)}$ that satisfies
  $A(i, q) = L(q)$ for all $q \in \binstr$.
\end{lemma}

\begin{proof}
  The one direction is obvious:
  If there exists a selector that works among $m(n)$ oracles,
  then letting $m(n) := 2$ yields a selector that works among two oracles.

  Conversely, let $S$ be a selector (that identifies an honest oracle among two oracles) with probability at least $1 - \frac{1}{3m(n)}$.
  Given an oracle $A$,
  let $A_i(q)$ denote $A(i, q)$ for any $i \in \N$.
  On input $x \in \binset^n$,
  we first make a query  $x$ to all the oracles $A_1, \cdots, A_{m(n)}$,
  and divide them into the two sets according to their answers:
  \begin{align*}
    C_0 &= \ISet{ j \in \NumSet{m(n)} } { A_j(x) = 0 }, \\
    C_1 &= \ISet{ k \in \NumSet{m(n)} } { A_k(x) = 1 }.
  \end{align*}
  \def\oracleanswer{\alpha}
  That is,
  $C_\oracleanswer \supplement{\oracleanswer \in \binset}$ is the set of the indices of all the oracles asserting that $L(x) = \oracleanswer$.

  Next, we repeat the following until $C_0 = \emptyset$ or $C_1 = \emptyset$:
  Pick arbitrary elements $j \in C_0$ and $k \in C_1$.
  We check which is a supposedly honest oracle by running
  $S^{A_j, A_k}$ on input $x$.
  If $S^{A_j, A_k}(x) = 0$, then we doubt $A_k$ and thus eliminate $k$ from $C_1$;
  Otherwise we doubt $A_j$ and eliminate $j$ from $C_0$.

  Finally, we output $1$ if and only if $C_1 \neq \emptyset$.

  Now let us analyze this algorithm.
  It runs in polynomial time because $|C_0| + |C_1|$ is decreased by one in each repetition.

  We claim the correctness of the algorithm.
  For simplicity, we assume that $L(x) = 0$.
  Then, there exists an index $i \in \NumSet{m(n)}$ such that $A_i$ is honest and $i \in C_0$.
  If $i \in C_0$ and some $k \in C_1$ are picked in a repetition,
  then $\Pr \left[ S^{A_i, A_k}(x) = 0 \right] \ge 1 - \frac{1}{3m(n)}$.
  That is, $i$ remains in $C_0$ with probability at least $1 - \frac{1}{3m(n)}$.
  Since $i$ is picked at most $|C_1| \supplement{\le m(n)}$ times,
  the probability that $i$ remains in $C_0$ is at least $1 - m(n) \cdot \frac{1}{3m(n)} = \frac 2 3$.
\end{proof}

\begin{remark}
  Although Lemma \ref{lemma:robustmany} is stated only for a probabilistic selector,
  analogous statements hold for a deterministic selector and a nonadaptive deterministic selector.
  For a deterministic selector, one can easily check that the same proof works.
  For a nonadaptive deterministic selector, we must compute the query set in polynomial time.
  On input $x$, let $Q(x)$ denote the query (to either $A_0$ or $A_1$) set of a selector that identifies an honest oracle among two oracles.
  Then we can define all the set of possible queries as $Q'(x) := \ISet { (i, q) \in \N \times \binstr } { 1 \le i \le m(|x|) , \  q \in Q(x) \cup \{x\} }$,
  which is clearly computable in polynomial time.
\end{remark}

By using Lemma \ref{lemma:robustmany},
we characterize the class of the paddable languages with a selector by
the property that
short advice can be removed under any relativized world.
In fact, we can prove a statement stronger than Theorem \ref{thm:characterizationelim}:
\begin{theorem}
  \label{thm:finercharact} \hspace{0em}
  \begin{enumerate}
    \item
      \label{enum:selector2elim}
      For any paddable language $L$,
      if $L$ has a selector, then
      $L \in \BPP^R \doubleslash \log$ implies $L \in \BPP^R$ for any oracle $R \subset \binstr$.
    \item
      \label{enum:elim2selector}
      For any language $L$,
      if $L \in \PTIME^R / 1$ implies $L \in \BPP^R$ for any oracle $R \subset \binstr$,
      then $L$ has a selector.
  \end{enumerate}
\end{theorem}
As a corollary, we immediately obtain Theorem \ref{thm:characterizationelim} (note that $\PTIME^R / 1 \subset \BPP^R \doubleslash \log$).

\begin{proof}
  \hspace{0em}
  \subparagraph*{Part \ref{enum:selector2elim}.}
  Let $M$ be a polynomial-time oracle machine which witnesses $L \in \BPP^R\doubleslash a$,
  where $a(n) = O(\log n)$.
  That is,
  there exists an advice function $\alpha \colon \binstr \to \binstr$ such that, for every $n \in \N$,
  \begin{align}
    \label{eq:bppadvicesecond}
    \Pr_{r \in \binset^{t(n)}} \left[ \forall x \in \binset^n,  \ M^R(x, r, \alpha(r)) = L(x) \right] \ge \frac56,
  \end{align}
  where $|\alpha(r)| = a(n)$ and $t$ is a polynomial (see also \cite[Definition 5.1]{TV07}).

  Let $l(n) \supplement{= n^{O(1)}}$ be an upper bound on the running time of a selector for $L$ on inputs of length $n$.
  By Lemma \ref{lemma:robustmany}, 
  there exists a selector $S$ that can identify an honest oracle among $m(n)$ oracles for $m(n) := 2^{a(l(n))} = n^{O(1)}$
  with probability at least $\frac56$.
  By padding, we may assume that $S$ makes only queries of length exactly $l(n)$ on each input length $n \in \N$
 
  Consider the following probabilistic algorithm:
  On input $x \in \binset^n$,
  pick a string $r \in_R \binset^{t(l(n))}$ uniformly at random,
  and define oracles by $A_i(q) := M^R(q, r, i)$ for any $q \in \binset^{l(n)}$,
  where $i \in \NumSet{m(n)}$ is identified with $i \in \binset^{a(l(n))}$.
  Simulate $S$ on input $x$, answering its queries $q \in \binset^{l(n)}$ to $A_i$ by computing $M^R(q, r, i)$.

  If a ``good'' string $r$ is picked (whose probability is at least $\frac56$ by \eqref{eq:bppadvicesecond}),
  then
  we have $A_i(q) = M^R(q, r, i) = L(q)$ for any $q \in \binset^{l(n)}$, where $i = \alpha(r)$.
  That is, $A_i$ is honest for some $i$ with probability at least $\frac56$.
  Thus,
  the algorithm computes $L$ correctly with probability at least $1 - \frac16 - \frac16 = \frac23$.

  \subparagraph*{Part \ref{enum:elim2selector}.}
  We prove the contraposition.
  Assume that $L$ does not have any selectors.

  Recall that
  we regard the computation given oracle access to two oracles $R_0, R_1$,
  namely $M^{R_0, R_1}$,
  as
  $M^R$ where $R(i \cdot q) = R_i(q)$ for each $i \in \binset$.
  Thus, the goal is to show that
  there exist oracles $R_0, R_1 \subset \binstr$ such that
  $L \in \PTIME^{R_0, R_1} / 1$ and $L \not\in \BPP^{R_0, R_1}$.

  We use a diagonalization argument on all the probabilistic polynomial-time oracle machine $M_1, M_2, \cdots$.
  We construct $R_0 \stage e , R_1 \stage e$ at stage $e \in \N$,
  and then define $R_i := \bigcup_e R_i \stage e$ for each $i \in \binset$.

  We will construct them so that,
  for each $n \in \N$,
  there exists $j_n \in \binset$ such that $R_{j_n}(q) = L(q)$ for any $q \in \binset^n$.
  Thus, $L \in \PTIME^{R_0, R_1} / 1$ holds because
  we can make a query $x$ to obtain $R_{j_n}(x) = L(x)$
  with advice $\{j_n\}_{n\in\N}$ of one bit.

  Let us now construct $R_0 \stage e, R_1 \stage e$, and $l \stage e \in \mathbb Z$,
  where $l \stage e$ represents the maximum length of the strings that have been fixed.
  At stage $e = 0$, we set $R_0 \stage 0 = R_1 \stage 0 = \emptyset$, and $l \stage 0 := -1$.

  At stage $e \ge 1$,
  we claim that $R_0 \stage {e-1}$ and $R_1 \stage {e-1}$ can be extended
  so that some input $x \stage e$ can fool $M_e$:
  \begin{claim}
    \label{claim:diagonalization}
    For each $e \ge 1$,
    there exist oracles $A_0, A_1 \subset \binstr$ and a string $x \stage e \in \binstr$  such that
    \begin{enumerate}
      \item
        \label{enum:consistentshort}
        $A_i$ agrees with $R_i \stage {e-1}$ on all the strings of length at most $l \stage {e-1}$
        for each $i \in \binset$,
      \item
        \label{enum:honestlong}
        either $A_0$ or $A_1$ agrees with $L$ on all the strings of length greater than $l \stage {e-1}$,
        and
      \item
        \label{enum:notselector}
        $\Pr \left[ M_e^{A_0, A_1} (x \stage e) = L(x \stage e) \right] < \frac 2 3$.
    \end{enumerate}
  \end{claim}

  \begin{proof}
    [Proof of Claim \ref{claim:diagonalization}]
    Assume otherwise.
    That is, for any oracles $A_0, A_1 \subset \binstr$ and string $x \in \binstr$,
    we have $\Pr \left[ M_e^{A_0, A_1} (x) = L(x) \right] \ge \frac 2 3$ if
    Properties \ref{enum:consistentshort} and \ref{enum:honestlong} hold.
    Then, the following algorithm yields a selector for $L$, which contradicts the assumption:
    we hardwire all the strings in $R_i \stage {e-1}$ of length at most $l \stage {e-1}$ into a table;
    given oracles $A_0, A_1$ one of which agrees with $L$,
    we simulate $M_e$, answering its queries $q$ to $A_i \supplement{ i \in \binset }$ with
    the content of the table if $|q| \le l \stage {e-1}$ and
    with $A_i(q)$ otherwise.
  \end{proof}

  Define $l \stage e \supplement{ > l \stage {e-1} }$ as an upper bound on
  the length of the queries that 
  $M_e^{A_0, A_1}( x \stage e ) $ makes.
  Then, define $R_i \stage {e}$ as
  $R_i \stage e (q) :=  R_i \stage {e-1} (q) = A_i (q)$ if $|q| \le l \stage {e-1}$;
  $R_i \stage e (q) :=   A_i (q)$ if $l \stage {e-1} < |q| \le l \stage{e}$;
  and
  $R_i \stage e (q) = 0$ otherwise,
  for each $q \in \binstr$.
  This completes the construction of stage $e$.

  On one hand,
  $x \stage e$ witnesses $M_e^{R_0, R_1}$ not computing $L$ on input $x \stage e$
  for any $e \ge 1$, by Property \ref{enum:notselector};
  thus, we have $L \not\in \BPP^{R_0, R_1}$.
  On the other hand,
  for each input length $n \in \N$,
  either $R_0$ or $R_1$ agrees with $L$ on $\binset^n$,
  by Property \ref{enum:honestlong}; thus, we have $L \in \PTIME^{R_0, R_1}/1$.
\end{proof}

\begin{remark}
  Again, the analogous statement (Theorem \ref{thm:characterizationdet}) holds for a deterministic selector.
  A proof is essentially the same and hence is omitted.

  One can also prove the quantitative version (Corollary \ref{cor:quantitative}) of Part \ref{enum:selector2elim} of Theorem \ref{thm:finercharact}
  by changing parameters in the proofs of Theorem \ref{thm:finercharact} and Lemma \ref{lemma:robustmany}.
\end{remark}

\section{Nonadaptive Deterministic Selector}
\label{sec:truth}
In this section we prove Theorem \ref{thm:honestytt}.

We first prove Part \ref{enum:honestyttlower} of Theorem \ref{thm:honestytt},
which states that \emph{every} $\PNP$-complete language has a nonadaptive deterministic selector.
It is sufficient to show that a \emph{specific} $\PNP$-complete language has a selector (recall Proposition \ref{prop:closedequiv} and Remark \ref{remark:closedequiv}).
We construct a nonadaptive deterministic selector for
the following canonical $\PNP$-complete language (see \cite{Kre88} for a proof of its completeness).

\begin{definition}[Lexicographically Maximum Satisfying Assignment; Krentel~\cite{Kre88}]
  The \emph{lexicographically maximum satisfying assignment problem} contains
  all the pairs $(\phi, k)$ such that $\phi \colon \binset^n \to \binset$ is a satisfiable Boolean formula in $n$ variables for some $n \in \N$,
  and
  $a_k = 1$, where $a_1 \cdots a_n \in \binset^n$ denotes the lexicographically maximum satisfying assignment of $\phi$.
\end{definition}

In other words, the lexicographically maximum satisfying assignment problem is
the decision version of the problem of answering, given a Boolean formula $\phi$ in $n$ variables,
the lexicographically maximum satisfying assignment if $\phi$ is satisfiable and $0^n$ otherwise.
Note that it is implicit in the definition that the answer is $0^n$ for an unsatisfiable Boolean formula.

\begin{proof}
  [Proof of Part \ref{enum:honestyttlower} of Theorem \ref{thm:honestytt}]
  We show an algorithm of a selector for the lexicographically maximum satisfying problem, together with its analysis.
  Let us call two oracles $A_0$ and $A_1$.

  On input $(\phi, k)$, the set of all the queries that we make is $\ESet{ (\phi, j) \mid j \in \NumSet n }$,
  where $n \in \N$ is the number of variables in $\phi$.
  The (presumably) lexicographically maximum satisfying assignment asserted by each oracle $A_i \ (i \in \binset)$
  can be obtained by concatenating the answers of the oracle, namely $A_i(\phi, 1) \cdot A_i(\phi, 2) \cdots A_i(\phi, n) =: \pa i \in \binset^n$.

  If the $k$th bits of $\pa0$ and $\pa1$ agree, then we simply output it because the oracles agree on input $(\phi, k)$.

  Otherwise $\pa0$ is not equal to $\pa1$.
  Therefore, we may assume without loss of generality that $\pa0 < \pa1$.
  We check whether $\pa1$ is a satisfying assignment or not by evaluating $\phi(\pa1)$.
  If $\phi(\pa1) = 1$, then we trust the oracle $A_1$ and output $A_1(\phi, k)$
  because $A_1$ showed a satisfying assignment larger than $\pa0$;
  otherwise we doubt $A_1$ and output $A_0(\phi, k)$
  because $A_1$ tried to cheat us by answering an unsatisfying assignment.
\end{proof}

Then we show that any language with a nonadaptive deterministic selector is in $\SH2$.
\begin{proof}
  [Proof of Part \ref{enum:honestyttupper} of Theorem \ref{thm:honestytt}]
  Let $L$ be a language with a nonadaptive deterministic selector $S$.
  We claim that $L$ is in $\SH2$.
  Let $Q(x) = \{q_1, \cdots, q_m\}$ be the query set of $S$ on input $x \in \binstr$.

  We consider the following polynomial-time machine $M$:
  Suppose that the input to $M$ is $(x, y, z) \in \binset^n \times \binset^m \times \binset^m$.
  Let $y = y_1 \cdots y_m$ and $z = z_1 \cdots z_m$.
  $M$ simulates the selector $S$ on input $x$.
  If $S$ makes a query $q_i$ to the oracle $A_0$,
  then it is answered with $y_i$.
  Similarly, if $S$ makes a query $q_i$ to the oracle $A_1$,
  then it is answered with $z_i$.

  Then, there exists $y \in \binset^m$ such that $M(x, y, z) = L(x)$ for any $z \in \binset^m$.
  Indeed,
  if $y$ is the concatenation of $L(q_1), \cdots, L(q_m)$,
  then by the definition of a nonadaptive deterministic selector, $M(x, y, z)$ correctly outputs $L(x)$ for any $z \in \binset^m$,
  because all the queries that $S$ makes to $A_0$ are answered correctly.
  Similarly, there exists $z \in \binset^m$ such that $M(x, y, z) = L(x)$ for any $y \in \binset^m$.
\end{proof}

\section{Probabilistic Selector}
\label{sec:bpp}

In this section we investigate probabilistic selectors.

First, we show that probabilistic selectors can be constructed based on instance checkers.
An instance checker is formally defined as follows:
\begin{definition}[Instance Checker~\cite{BK95}]
  An \emph{instance checker} $C$ for a language $L$ is a probabilistic polynomial-time oracle machine such that,
  given any oracle $A \subset \binstr$,
  \begin{enumerate}
    \item
      if $A = L$ then $C^A$ accepts with high probability, \ie $\Pr \left[ C^A(x) = 1 \right] \ge \frac23$ on all the input $x \in \binstr$, and
    \item
      for any input $x \in \binstr$, if $A(x) \neq L(x)$ then $C^A(x)$ rejects with high probability, \ie
      $\Pr \left[ C^A(x) = 0 \right] \ge \frac23$,
  \end{enumerate}
  where the probability is taken over coin flips of $C$.
\end{definition}

\begin{proposition}
  Every language with an instance checker has a selector.
\end{proposition}

\begin{proof}
  Suppose that a language $L$ has an instance checker $C$.
  Given input $x \in \binstr$ and two oracles $A_0, A_1 \subset \binstr$, we check which is honest, $A_0$ or $A_1$,
  by computing $C^{A_0}(x)$.
  If $C^{A_0}(x)$ accepts, then we trust $A_0$ and output $A_0(x)$; otherwise we doubt $A_0$ and output $A_1(x)$.

  Let us analyze the algorithm above.
  If $A_0 = L$, then $C^{A_0}(x)$ accepts with probability at least $\frac23$, and hence we can output $A_0(x) = L(x)$ correctly
  with probability at least $\frac23$.

  Otherwise, it must hold that $A_1 = L$.
  If $A_0(x) = L(x)$, then we can surely output $L(x)$ correctly since $A_0(x) = A_1(x) = L(x)$.
  If $A_0(x) \neq L(x)$, then $C^{A_0}(x)$ rejects with probability at least $\frac23$,
  and thus we can output $A_1(x) = L(x)$ correctly with probability at least $\frac23$.
\end{proof}

Next, we show an upper bound on the languages with a probabilistic selector.
For completeness, we include a definition of $\SHEXP$,
which is a straightforward exponential-time analogue of $\SH2$:
\begin{definition}
  \label{def:shexp}
  We say that a language $L$ is in $\SHEXP$ if there exist
  a time-constructible function $t(n) = 2^{n^{O(1)}}$ and
  a Turing machine $M$ running in time $2^{|x|^{O(1)}}$ on input $(x, \cdot, \cdot)$
  such that, for any input $x \in \binstr$,
  \begin{align*}
    \exists y \in \binset^{t(|x|)}, \forall z \in \binset^{t(|x|)}, \  M(x, y, z) = L(x), \\
    \exists z \in \binset^{t(|x|)}, \forall y \in \binset^{t(|x|)}, \  M(x, y, z) = L(x).
  \end{align*}
\end{definition}
The proof itself is essentially a corollary of Part \ref{enum:honestyttupper} of Theorem \ref{thm:honestytt}:
\begin{proof}[Proof of Part \ref{enum:honestyupper} of Theorem \ref{thm:honesty}]
  Notice that a probabilistic selector can be simulated by an exponential-time nonadaptive deterministic selector.
  In addition, every language with an exponential-time nonadaptive deterministic selector is in $\SHEXP$,
  which is an exponential-time analogue of Part \ref{enum:honestyttupper} of Theorem \ref{thm:honestytt}.
  Combining these two facts, it follows that every language with a probabilistic selector is in $\SHEXP$.
\end{proof}

\subsection{Selector for $\EXPNP$-complete Languages}
In this subsection we prove the main theorem (Theorem \ref{thm:honesty}, Part \ref{enum:honestylower}).
That is, we construct a selector for $\EXPNP$-complete languages.

\subsection*{Proof Sketch}
We sketch the proof of the main theorem.
We will construct a selector for a specific $\EXPNP$-complete language, which is
a problem of finding
the lexicographically maximum satisfying assignment
of a succinctly described Boolean formula $F_\Phi \colon \binset^{2^n} \to \binset$.
The basic strategy to construct a selector for this language is
the same with that of Part \ref{enum:honestyttlower} of Theorem~\ref{thm:honestytt}:
Given access to two oracles $A_0, A_1 \subset \binstr$,
we request them to reveal the presumably lexicographically maximum satisfying assignments $V_0, V_1 \in \binset^{2^n}$ asserted by $A_0, A_1$, respectively.
The rest of the algorithm consists of two parts:
First, we determine the larger assignment of $V_0$ and $V_1$, checking whether $V_0 < V_1$ or $V_0 > V_1$.
Second, we verify whether the larger assignment satisfies the formula $F_\Phi$ or not.
Obviously, the obstacle is that there can be exponentially many variables and clauses in $F_\Phi$.

For the second part,
Babai, Fortnow, and Lund~\cite{BFL91} showed that, given access to provers (or, equivalently, an oracle),
one can efficiently check that exponentially many constraints in $F_\Phi$ are satisfied:
basically,
by encoding an assignment as a multilinear function and using arithmetization,
it holds that the assignment satisfies all the clauses in $F_\Phi$  if and only if the sum of some low-degree polynomials (that can be computed by the multilinear function and the arithmetization) over a subdomain $\binset^l$ is equal to $0$,
and the latter can be verified by using the sum-check protocol~\cite{LFKN92} (called the LFKN protocol in \cite{BFL91}).
As pointed out by G\'abor Tardos~\cite{BFL91},
since $\EXPNP$ is capable of finding a satisfying assignment of an exponential-sized Boolean formula, the honest oracle in the protocol above can be implemented in $\EXPNP$;
thus, given access to an honest $\EXPNP$-complete oracle (which is $A_0$ or $A_1$), one can verify the satisfiability.

For the first part, we perform a binary search to obtain the lexicographically first index $z$
such that $V_0$ and $V_1$ disagree.
Thus, we need
\begin{enumerate}
  \item
    \label{enum:binaryeq}
    to check if $V_0 = V_1$ on some range of indices, and
  \item
    \label{enum:binarysplit}
    to split the range into two parts.
\end{enumerate}
We observe that these can be done if we encode a satisfying assignment by the multilinear extension (as with \cite{BFL91}):
Let $\mathbb F$ be a finite field.
We regard the assignments $V_0, V_1 \in \binset^{2^n}$ as vectors in $\mathbb F^{2^n}$.
There is a bijective correspondence between a vector $V \in \mathbb F^{2^n}$ and a multilinear function $\widetilde V \colon \mathbb F^n \to \mathbb F$.
For example, if $n = 2$ and $V = (V_{00}, V_{01}, V_{10}, V_{11})$,
then
\[
  \widetilde V (x_1, x_2) = V_{00} (1 - x_1) (1 - x_2) + V_{01} (1 - x_1) x_2 + V_{10} x_1 (1 - x_2) + V_{11} x_1 x_2.
\]

For Part \ref{enum:binaryeq}, we can rely on the polynomial identity testing:
indeed, since the multilinear extension is bijective, we have $V_0 \neq V_1$ if and only if these multilinear extensions $\widetilde V_0$ and $\widetilde V_1$ differ;
thus, it is sufficient to check if the two low-degree polynomials $\widetilde V_0$ and $\widetilde V_1$ differ.

It is well known that, given access to two low-degree polynomials, one can efficiently check if these polynomials differ:
given access to two functions $\widetilde V_0, \widetilde V_1$,
pick a random point $u \in_R \mathbb F^n$ and check if $\widetilde V_0(u) \neq \widetilde V_1(u)$.
Assuming that the functions are low-degree (which is true if they are multilinear),
the Schwartz-Zippel lemma assures that $\widetilde V_0$ and $\widetilde V_1$ disagree on a large fraction of inputs if $\widetilde V_0 \neq \widetilde V_1$.
Although it is possible that a dishonest oracle tries to cheat us by storing a high-degree polynomial,
we can check whether or not the function stored by an oracle is close to some multilinear function,
by using the multilinearity test~\cite{BFL91}.

For Part \ref{enum:binarysplit}, we use the following simple fact:
Fixing the first variable of a multilinear extension $\widetilde V$ to $0$ or $1$,
we obtain multilinear extensions that correspond to the first or second part of $V$.
In the example above, we obtain two multilinear functions:
\begin{align*}
  \widetilde V(0, x_2) = V_{00} (1 - x_2) + V_{01} x_2, \quad   \widetilde V(1, x_2) = V_{10} (1 - x_2) + V_{11} x_2.
\end{align*}
These correspond to multilinear extensions of $(V_{00}, V_{01})$ and $(V_{10}, V_{11})$, respectively, for $n = 1$.
Thus, we can recursively compute the lexicographically first disagreement.

\subsection*{Proof of the Main Theorem}
Now we move on to the proof of the main theorem.
We construct a selector for the following $\EXPNP$-complete language,
which is an analogue of the $\NEXP$-complete languages called the oracle-3-satisfiability problem in \cite{BFL91}.
\begin{definition}
  [Lexicographically Maximum Oracle-3-satisfying Assignment]
  Let $m, n$ be nonnegative integers, and $\Phi \colon \binset^{m+3n+3} \to \binset$ be a Boolean formula.
  For a Boolean function $X \colon \binset^n \to \binset$,
  define $\satinstance \Phi X$ as the following Boolean formula:
  \begin{align*}
    \bigwedge_{w \in \binset^{m+3n}} \Phi(w, X(b_1), X(b_2), X(b_3)),
  \end{align*}
  where $w = (y, (b_1, b_2, b_3)) \in \binset^m \times \left( \binset^n \right)^3$.
  A Boolean function $X \colon \binset^n \to \binset$ is said to be an assignment of $F_\Phi$.
  For assignments $X, Y \colon \binset^n \to \binset$, we introduce the lexicographical ordering:
  $X$ is less than $Y$ if there exists an index $b \in \binset^n$ such that $X(b) < Y(b)$ and $X(b') = Y(b')$ for any $b' < b$.
  Let $\lexass \colon \binset^n \to \binset$ denote the lexicographically maximum assignment such that $\satinstance \Phi \lexass = 1$
  (\ie the lexicographically maximum satisfying assignment of $F_\Phi$);
  if there is no satisfying assignment, then define $\lexass(b) = 0$ for any $b \in \binset^n$.

  The \emph{lexicographically maximum oracle-3-satisfying assignment} is a problem
  of answering $\lexass(\binput)$, given 
  nonnegative integers $m, n$, a Boolean formula $\Phi \colon \binset^{m+3n+3} \to \binset$, and an index $\binput \in \binset^n$ as input.
\end{definition}
We omit a proof of $\EXPNP$-completeness because this is a simple exponential-time analogue of the lexicographically maximum satisfying assignment
language~\cite{Kre88} (see also \cite{BFL91}).

Suppose that the input is a Boolean formula $\Phi \colon \binset^{m+3n+3} \to \binset$ and an index $\binput$, and that we have access to two oracles $A_0$ and $A_1$,
one of which is honest.

\subsection*{Encoding Assignments by the Multilinear Extension}
As with the proof of $\MIP = \NEXP$~\cite{BFL91}, we encode a satisfying assignment by the multilinear extension.
Let $\mathbb F$ be a prime field such that $|\mathbb F|$ is sufficiently large (but is bounded by a polynomial in the input size).
We regard $\binset \subset \mathbb F$ in the canonical way.
We say that a function $f \colon \mathbb F^n \to \mathbb F$ is multilinear if it is a polynomial of degree at most $1$ in each variable.
\begin{proposition}
  [Multilinear Extension]
  \label{prop:multilinearextension}
  Let $f \colon \binset^n \to \mathbb F$ be an arbitrary function.
  Then, there exists a unique multilinear function $\widetilde f \colon \mathbb F^n \to \mathbb F$ such that
  $f$ and $\widetilde f$ agree on $\binset^n$.
\end{proposition}
\begin{proof}
  [Proof Sketch]
  For a complete proof, the reader is referred to \cite[Proposition 4.4]{BFL91}.
  Here, we note that the extension $\widetilde f$ can be explicitly written as
  \begin{align}
    \label{eq:extension}
    \widetilde f(x) = \sum_{b \in \binset^n} f(b) \prod_{i = 1}^n \left( (1 - x_i) (1 - b_i) + x_i b_i \right),
  \end{align}
  where $b = (b_1, \cdots, b_n)$ and $x = (x_1, \cdots, x_n) \in \mathbb F^n$.
\end{proof}

For the lexicographically maximum satisfying assignment $\lexass \colon \binset^n \to \binset \subset \mathbb F$,
let $\widetilde\lexass \colon \mathbb F^n \to \mathbb F$ denote its multilinear extension.

We request the oracles to grant local access to $\widetilde\lexass$.
Formally, we consider the following search problem:
given a Boolean formula $\Phi$, a prime $|\mathbb F|$, and $x \in \mathbb F^n$,
the task is to output the value $\widetilde \lexass(x)$.
We regard this problem as a decision problem in the standard way.
(Specifically, given the inputs specified above and auxiliary inputs $k \in \N$ and $b \in \binset$,
the task is to output one bit saying whether or not the $k$th bit of a binary representation of $\widetilde \lexass(x)$ is $b$.)
The problem is still solvable in $\EXPNP$,
by first computing $V_\Phi$ in $\EXPNP$ and then computing the expression \eqref{eq:extension} straightforwardly in exponential time.

Therefore, the problem can be reduced to the original $\EXPNP$-complete problem;
by using the $\EXPNP$-completeness,
one can translate the problem of computing $\widetilde \lexass(x)$ into the original problem in polynomial time,
and hence we can ask the oracles to output $\widetilde \lexass(x)$.
Let $f_0, f_1 \colon \mathbb F^n \to \mathbb F$ denote the answers of the oracles $A_0, A_1$, respectively.
Then, we have $f_i = \widetilde \lexass$ for an honest oracle $A_i$.

Although $f_i$ is not necessarily multilinear for a dishonest oracle $A_i$,
we can ensure that it is close to some multilinear function.
This can be done by the multilinearity test, which was one of the main technical ingredients in the proof of $\MIP = \NEXP$~\cite{BFL91}.
For two functions $f, g \colon \mathbb F^n \to \mathbb F$ and a real number $\delta \in \mathbb R$,
we say that $f$ and $g$ are $\delta$-close if $\Pr_{x \in \mathbb F^n} [ f(x) \neq g(x) ] < \delta$.
\begin{lemma}
  [Multilinearity Test \cite{BFL91}]
  Let $n \in \N$ and $\mathbb F$ be a finite field.
  There exist a constant $\delta = n^{O(1)}/|\mathbb F|$ and  an efficient probabilistic algorithm that, given oracle access to an arbitrary function $f \colon \mathbb F^n \to \mathbb F$,
  \begin{enumerate}
    \item
      accepts with probability $1$ if $f$ is multilinear, and
    \item
      rejects with high probability if $f$ is not $\delta$-close to any multilinear function.
  \end{enumerate}
\end{lemma}
We perform the multilinearity test for $f_0$ and $f_1$.
Suppose that $f_i$ is not $\delta$-close to any multilinear function for a dishonest oracle $A_i$.
Then, the multilinearity test fails and hence we can doubt $A_i$ with high probability.
Therefore, in what follows, we may assume that both $f_0$ and $f_1$ are $\delta$-close to some multilinear functions $\hat f_0$ and $\hat f_1$, respectively
(note that $\hat f_0$ and $\hat f_1$ are unique for small $\delta$).

In reality, we have only access to $f_0, f_1$ instead of multilinear functions $\hat f_0, \hat f_1$.
However, we may pretend to have access to the multilinear functions $\hat f_0, \hat f_1$,
by using the random self-reducibility of multivariate low-degree polynomials (also known as the self-correction of the Reed-Muller code).
\begin{lemma}
  [Self-correction; Beaver and Feigenbaum~\cite{BF90} and Lipton~\cite{Lip91}]
  There exists an efficient probabilistic algorithm that,
  given input $x \in \mathbb F^n$ and oracle access to a function $f \colon \mathbb F^n \to \mathbb F$ that is $\delta$-close to a multilinear function $\hat f \colon \mathbb F^n \to \mathbb F$,
  outputs $\hat f(x)$ with probability at least $1 - \delta (n+1)$.
\end{lemma}
\begin{proof}
  Let $a_0, \cdots, a_n$ be arbitrary distinct points in $\mathbb F \setminus \{0\}$.
  Pick a random point $y \in_R \mathbb F^n$.
  By the polynomial interpolation, find the univariate polynomial $p$ of degree at most $n$ such that $f(x + a_i \cdot y) = p(a_i)$ for all $i \in \{0, \cdots, n\}$,
  and output $p(0)$.

  Since $x + a_i \cdot y$ is uniformly distributed on $\mathbb F^n$ for any fixed $x$ and $a_i \neq 0$,
  it holds that $\hat f(x + a_i \cdot y) = f(x + a_i \cdot y)$ with probability at least $1 - \delta$.
  By the union bound, we have $p(a_i) = \hat f(x + a_i \cdot y)$ for each $i \in \{0, \cdots, n\}$ with probability at least $1 - \delta (n+1)$;
  thus we have $p(0) = \hat f(x)$ with probability at least $1 - \delta (n+1)$, because $\hat f$ is a polynomial of total degree at most $n$.
\end{proof}
\begin{remark}
  In the case of the proof of $\MIP = \NEXP$, the self-correcting algorithm was not needed;
  for the sum-check protocol, it is sufficient to evaluate a multilinear function $\hat f_i$ on random points $x \in_R \mathbb F^n$, rather than fixed points.
  In contrast, we need to evaluate a multilinear function $\hat f_i$ on points that are not uniformly distributed, during the binary search.
\end{remark}

In the following, we pretend that the dishonest oracle $A_i$ asserts that the satisfying assignment is $\restr{\hat f_i}{\binset^n}$, instead of $\restr{f_i}{\binset^n}$.
(Note that it holds that $\restr{f_i}{\binset^n} = \restr{\hat f_i}{\binset^n} = \lexass$ for the honest oracle $A_i$.)

\subsection*{Identifying the Larger Assignment}

We are now ready to describe how to identify the larger assignment.
It is sufficient to show that
we can find, with high probability, the lexicographically first index $z \in \binset^n$
such that $\hat f_0(z) \neq \hat f_1(z)$.

First, we check if $\hat f_0(\binput) = \hat f_1(\binput)$:
For each $i \in \binset$, run the self-correcting algorithm for $f_i$ to obtain $\hat f_i(\binput)$.
If $\hat f_0(\binput) = \hat f_1(\binput)$, then output it (which is surely the correct answer since $\hat f_i(\binput) = \lexass(\binput)$ for the honest oracle $A_i$) and halt.
Otherwise, perform the binary search described below.

We compute the lexicographically first disagreement $z = (z_1, \cdots, z_n) \in \binset^n$ one by one.
For $j := 1$ to $n$, repeat the following:
Suppose that we have computed $z_1, \cdots, z_{j-1}$.
Pick a random point $u = (u_{j+1}, \cdots, u_n) \in_R \mathbb F^{n-j}$ uniformly at random.
Define $x := (z_1, \cdots, z_{j-1}, 0, u_{j+1}, \cdots, u_n) \in \mathbb F^n$.
For each $i \in \binset$, use the self-correcting algorithm for $f_i$ to obtain $\hat f_i (x)$.
If $\hat f_0(x) \neq \hat f_1(x)$, then set $z_j := 0$; else, set $z_j := 1$.

\begin{claim}
  Assume that $\hat f_0(\binput) \neq \hat f_1(\binput)$.
  Let $z \in \binset^n$ denote the lexicographically first index such that $\hat f_0(z) \neq \hat f_1(z)$.
  Then, the binary search described above correctly computes $z$ with probability at least $1 - \delta n(n+1) - \frac{n^2}{|\mathbb F|}$.
\end{claim}
In particular, by setting $|\mathbb F|$ large enough, we can compute $z$ with high probability.
\begin{proof}
  Let $j \in \{1, \cdots, n\}$.
  Consider the $j$th iteration and assume that we have computed $z_1, \cdots, z_{j-1}$ correctly.
  For each $i \in \binset$,
  let $f'_i \colon \mathbb F^{n-j} \to \mathbb F$ be the multilinear function
  such that
  \[
    f'_i(t_{j+1}, \cdots, t_n) = \hat f_i(z_1, \cdots, z_{j-1}, 0, t_{j+1}, \cdots, t_n),
  \]
  for any $(t_{j+1}, \cdots, t_n) \in \mathbb F^{n-j}$.
  (The binary search tries to check if $f'_0 \neq f'_1$ by the polynomial identity testing, and sets $z_j := 0$ if and only if $f'_0 \neq f'_1$.)

  If $z_j = 0$, then we have $f'_0 \neq f'_1$ because $f'_0(z_{j+1}, \cdots, z_n) \neq f'_1(z_{j+1}, \cdots, z_n)$.
  The probability that the self-correcting algorithm outputs $\hat f_i(x)$ correctly is at least $1 - \delta(n+1)$ for a dishonest oracle $A_i$.
  By the Schwartz-Zippel lemma,
  the probability that $f'_0(u) \neq f'_1(u)$ for a random point $u \in_R \mathbb F^{n-j}$ is at least $1 - \frac{n-j}{|\mathbb F|} \ge 1 - \frac{n}{|\mathbb F|}$.
  Therefore, the algorithm sets $z_j := 0$ correctly with probability at least $1 - \delta(n+1) - \frac{n}{|\mathbb F|}$.

  If $z_j = 1$, then
  it follows from the minimality of $z$ that
  $f'_0(t) = f'_1(t)$ for every $t \in \binset^{n-j}$.
  Since $f'_0$ and $f'_1$ are multilinear, we have $f'_0 = f'_1$
  by the uniqueness of the multilinear extension (Proposition \ref{prop:multilinearextension})
  and hence $f'_0(u) = f'_1(u)$ holds for any $u \in \mathbb F^{n-j}$.
  Therefore, since the self-correcting algorithm outputs $\hat f_i(x)$ with probability at least $1 - \delta(n+1)$,
  the algorithm sets $z_j := 1$ correctly with probability at least $1 - \delta(n+1)$.

  Overall, the algorithm computes $z$ correctly with probability at least
  \[
    \left( 1 - \delta(n+1) - \frac{n}{|\mathbb F|} \right)^n \ge 1 - \delta n (n+1) - \frac{n^2}{|\mathbb F|}.
  \]
\end{proof}

We have computed the lexicographically first disagreement $z \in \binset$ such that $\hat f_0(z) \neq \hat f_1(z)$.
Run the self-correcting algorithm to obtain $\hat f_0(z)$ and $\hat f_1(z)$.
Without loss of generality (by swapping the oracles if $\hat f_0(z) > \hat f_1(z)$),
we may assume that $\hat f_0(z) < \hat f_1(z)$.

Now we know, with high probability, that $A_1$ asserts the larger (presumably satisfying) assignment $\restr{\hat f_1}{\binset^n} \colon \binset^n \to \mathbb F$.

\subsection*{Verifying the Satisfiability}
All that remains is to verify that $\restr{\hat f_1}{\binset^n}$ satisfies $F_\Phi$,
which can be done in the same way with a proof of $\MIP = \NEXP$.
For completeness, we sketch a proof suggested in \cite[Section 7.1]{BFL91} and 
observe that it can be done with the help of an $\EXPNP$-complete oracle.

Babai, Fortnow, Lund~\cite{BFL91} used the sum-check protocol~\cite{LFKN92} to check whether or not
an exponentially long assignment satisfies $F_\Phi$.
Basically, checking if an assignment $\restr{\hat f_1}{\binset^n} \colon \allowbreak \binset^n \to \mathbb F$ satisfies a Boolean formula $F_\Phi$
reduces to checking if some low-degree polynomials $g \colon \mathbb F^l \to \mathbb F$ evaluate to $0$ on $\binset^l$.

Let us arithmetize the Boolean formula $\Phi \colon \binset^{m+3n+3} \to \binset$ to a low-degree polynomial $\widetilde\Phi \colon \mathbb F^{m+3n+3} \to \mathbb F$ in the standard way,
so that $\Phi$ and $\widetilde\Phi$ agree on $\binset^{m+3n+3}$ (see \cite[Section 3.1]{BFL91}).
Define $g^1 \colon \mathbb F^{m+3n} \to \mathbb F$ and $g^2 \colon \mathbb F^n \to \mathbb F$ as
\begin{align}
  \label{eq:g1}
  g^1(w) &:= 1 - \widetilde\Phi\left(w, \extass(b_1), \extass(b_2), \extass(b_3) \right),
  \\
  \label{eq:g2}
  g^2(b) &:= \extass(b) \left( 1 - \extass(b) \right),
\end{align}
where $w = \left(y, (b_1, b_2, b_3) \right) \in \mathbb F^m \times \left(\mathbb F^n \right)^3$ and $b \in \mathbb F^n$.
Note that since $\extass$ and $\widetilde\Phi$ are low-degree polynomials, so are $g^1$ and $g^2$.

It is easy to see that $g^1(w) = 0$ and $g^2(b) = 0$ for any $w \in \binset^{m+3n}$ and $b \in \binset^n$ if and only if $\restr{\extass}{\binset^n}$ is a satisfying assignment of $F_\Phi$.
Indeed, $g^2(b) = 0$ forces $\restr{\extass}{\binset^n}$ to be a Boolean function (\ie $\extass(b) \in \binset$ for any $b \in \binset^n$),
and $g^1(w) = 0$ means that $\Phi(w, \extass(b_1), \extass(b_2), \extass(b_3))$ is true for any $w \in \binset^{m+3n}$.

We note that, given a random point $w$ or $b$, we can compute the value of $g^1(w)$ or $g^2(b)$ with high probability by substituting $f_1$ for $\extass$ in \eqref{eq:g1} or \eqref{eq:g2}
(\ie we do not need to use the self-correcting algorithm);
for a random point $w \in_R \mathbb F^{m+3n}$, it holds that $g^1(w)$ computed by substituting $f_1$ in \eqref{eq:g1} and $g^1(w)$ are identical with probability at least $1 - 3 \delta$.

Therefore, it is sufficient to show that we can check if each $g \in \{g^1, g^2\}$ vanishes on $\binset^l$,
given access to a low-degree polynomial $g$.
(Here, $l := m + 3n$ if $g = g^1$ and $l := n$ if $g = g^2$.)
There are several ways to verify that $g \colon \mathbb F^l \to \mathbb F$ vanishes on $\binset^l$, including \cite[Section 7.1]{BFL91} and \cite{BFLS91,FGLSS96,BS08}.
Here, we follow the way of Feige, Goldwasser, Lov\'{a}sz, Safra, and Szegedy~\cite{FGLSS96}.

We reduce a task of checking if $g \colon \mathbb F^l \to \mathbb F$ vanishes on $\binset^l$ to a task of checking if a sum is equal to $0$, the latter of which can be verified by the sum-check protocol (see \cite[Section 4.2.2]{FGLSS96} for more details):
Pick a random point $t = (t_1, \cdots, t_l) \in_R \mathbb F^l$.
Consider the following sum:
\begin{align}
  \label{eq:sumg}
  \sum_{w = (w_1, \cdots, w_l) \in \binset^l} g(w) \prod_{\{i \mid w_i = 1\} } t_i = \sum_{w \in \binset^l} g(w) \prod_{i \in \{1, \cdots, l\}} (w_i t_i + 1 - w_i).
\end{align}
If $g$ vanishes on $\binset^l$, then this sum is equal to $0$.
Otherwise, regarding the left-hand side of \eqref{eq:sumg} as a multilinear function on variables $t_1, \cdots, t_l$,
the sum is not equal to $0$ with probability at least $1 - \frac{l}{|\mathbb F|}$
by the Schwartz-Zippel lemma.
Therefore, by defining a low-degree polynomial $h_t \colon \mathbb F^l \to \mathbb F$ as $h_t(w) := g(w) \prod_{i \in \{1, \cdots, l\}} (w_i t_i + 1 - w_i)$ for any $w \in \mathbb F^l$, it is sufficient to check if
the sum of $h_t(w)$ over $w \in \binset^l$ is equal to $0$, which can be done by the sum-check protocol.

We describe the sum-check protocol briefly
(see \cite[Section 3.2]{BFL91} for a detailed description):
In order to check if $\sum_{w \in \binset^l} h_t(w) = 0$, pick a random point $r = (r_1, \cdots, r_l) \in_R \mathbb F^l$.
Define a low-degree univariate polynomial $g_i \colon \mathbb F \to \mathbb F$ for each $i \in \{1, \cdots, l\}$ as
\begin{align*}
  g_i(x) := \sum_{(w_{i+1}, \cdots, w_l) \in \binset^{l-i}} h_t(r_1, \cdots, r_{i-1}, x, w_{i+1}, \cdots, w_l)
\end{align*}
and $g_0(x) := 0$.
We request the oracle $A_1$ to reveal all the coefficients of the univariate polynomial $g_i$ for all $i \in \{1, \cdots, l\}$.
We trust $A_1$ if and only if
$g_{i-1}(r_{i-1}) = g_i(0) + g_i(1)$ for each $i \in \{1, \cdots, l\}$ (\emph{the Consistency Test}) and
$g_l(r_l) = h_t(r)$ (\emph{the Final Test}).
Here, since $r$ is a random point, we may evaluate $h_t(r)$ by using $f_1$ in place of $\extass$ in \eqref{eq:g1} and \eqref{eq:g2}.

We claim that the complexity of the honest oracle to output $g_i$ is bounded by $\EXPNP$.
Consider the following search problem:
given a Boolean formula $\Phi$, a prime $|\mathbb F|$, and $r, t \in \mathbb F^l$,
the task is to output all the coefficients of $g_i$ for all $i \in \{1, \cdots, l\}$
(which can be written in a binary representation of polynomial length),
where $\widetilde\lexass$ is substituted for $\extass$ in \eqref{eq:g1} and \eqref{eq:g2}.
Regarding this problem as a decision problem,
one can easily show that the problem is computable in $\EXPNP$.
Thus, we can request the oracle $A_1$ to output $g_i$.

Finally, we conclude the proof by analyzing the correctness (assuming that the binary search succeeded):
\begin{enumerate}
  \item
    If $A_1$ is honest, then $\extass = f_1 = \widetilde\lexass$.
    Thus, each $g \in \{g^1, g^2\}$ vanishes on $\{0, 1\}^l$, and hence the sum \eqref{eq:sumg} is $0$;
    therefore, we can trust $A_1$ with probability $1$.
  \item
    If $A_1$ is dishonest, then $\extass$ does not constitute a satisfying assignment of $F_\Phi$.
    (If it were a satisfying assignment, then $\restr{\hat f_1}{\binset^n}$ would be a satisfying assignment larger than $\restr{\hat f_0}{\binset^n} = \lexass$.)
    Thus, for some $g \in \{g^1, g^2\}$, the sum \eqref{eq:sumg} is not $0$ with probability at least $1 - \frac{l}{|\mathbb F|}$.

    Assume that the sum is not $0$, and let $d \in \N$ be an upper bound on the degree of the low-degree polynomial $h_t$.
    Suppose that the dishonest oracle claimed that $g_i$ is $g'_i$ for each $i \in \{1, \cdots, l\}$.
    Assuming that the Consistency Tests pass (\ie $g'_{i-1}(r_{i-1}) = g'_i(0) + g'_i(1)$ for each $i \in \{1, \cdots l\}$),
    it holds that $g'_l (r_l) \neq g_l (r_l) = h_t(r)$ with probability at least $1 - \frac{dl}{|\mathbb F|}$ (see \cite[Section 3.2]{BFL91}).
    The probability that $h_t$ can be evaluated correctly on a random point $r \in_R \mathbb F^l$ is at least $1 - 3 \delta$.
    Thus, the Final Test (\ie $g'_l (r_l) = h_t(r)$) fails with probability at least $1 - \frac{dl}{|\mathbb F|} - 3\delta$.

    Overall, we can doubt $A_1$ with probability at least $1 - \frac{dl}{|\mathbb F|} - 3 \delta - \frac{l}{|\mathbb F|}$.
\end{enumerate}

\section{Deterministic Selector}
\label{sec:turing}
This section is devoted to investigating a deterministic selector.

To prove the existence of a deterministic selector for a $\PSPACE$-complete language (Theorem \ref{thm:honestyT}, Part \ref{enum:honestyTlower}),
we show that a deterministic selector can be constructed based on downward self-reducibility:
\begin{theorem}
  Any downward self-reducible language has a deterministic selector.
\end{theorem}
Since there exists a downward self-reducible $\PSPACE$-complete language,
we immediately obtain a deterministic selector for any $\PSPACE$-complete language.
\begin{proof}
  Let $L$ be a downward self-reducible language.
  Namely, there exists a polynomial-time oracle machine $M$ such that
  \begin{itemize}
    \item
      $M^L(x) = L(x)$ for any $x \in \binstr$, and
    \item
      $M$ does not make any queries of length greater than or equal to $|x|$,
      on input $x \in \binstr$.
  \end{itemize}

  The idea is to keep a string $y$ such that $A_0(y) \neq A_1(y)$, and to run $M^{A_0}$ and $M^{A_1}$ to obtain another string $q$ of length less than $|y|$
  such that $A_0(q) \neq A_1(q)$.
  Consider the following algorithm:
  Given an input $x \in \binstr$ and two oracles $A_0, A_1$,
  if $A_0(x) = A_1(x)$ then output it and halt.
  Else, let $y := x$ and repeat the following:
  Compute $M^{A_i}(y)$ for each $i \in \binset$.
  If $M^{A_0}(y) = M^{A_1}(y) =: b$, then
  we trust the oracle $A_i$ such that $A_i(y) = b$ and output $A_i(x)$.
  Otherwise, let $q$ be the first query that $M^{A_0}$ and $M^{A_1}$ make on input $y$ such that $A_0(q) \neq A_1(q)$.
  (There exists such a $q$ because $M^{A_0}(y) \neq M^{A_1}(y)$; moreover, it holds that $|q| < |y|$ by the definition of downward self-reducibility.)
  Then, we update $y := q$ and move on to the next iteration.

  This algorithm runs in polynomial time, since $|y|$ decreases in each repetition.

  We claim the correctness of the algorithm.
  It is easy to see that $A_0(y) \neq A_1(y)$  at the beginning of each repetition.
  Suppose that $M^{A_0}(y) = M^{A_1}(y) =: b$.
  Since $A_0$ or $A_1$ is equal to $L$,
  we have $b = M^{A_0}(y) = M^{A_1}(y) = M^L(y) = L(y)$,
  where the last equality holds by the definition of $M$.
  Moreover, there exists the unique $i \in \binset$ such that $A_i(y) = b$ because $A_0$ and $A_1$ disagree on $y$.
  Therefore, $A_i$ is honest if and only if $A_i(y) = b \supplement{= L(y)}$.
\end{proof}

Then, we claim that any language with a deterministic selector is in $\PSPACE$ (Theorem \ref{thm:honestyT}, Part \ref{enum:honestyTupper}).
We thereby prove that the supremum of the languages with a deterministic selector is $\PSPACE$.
\begin{proof}[Proof of Part \ref{enum:honestyTupper} of Theorem \ref{thm:honestyT}]
  Let $L$ be a language with a deterministic selector $S$.

  The idea is to regard a computation of $S$ as a game played between the NO player and the YES player (which correspond to two oracles $A_0$ and $A_1$, respectively):
  On input $x \in \binstr$,
  the YES player tries to convince the selector $S$ that $x \in L$, whereas the NO player tries to convince $S$ that $x \not\in L$.
  The YES player chooses $A_1 \subset \binstr$ such that $x \in A_1$,
  and the NO player chooses $A_0 \subset \binstr$ such that $x \not\in A_0$.
  Then, we simulate $S^{A_0, A_1}(x)$, and the YES player wins if and only if $S^{A_0, A_1}(x) = 1$.

  It is easy to see that the YES player has a winning strategy if $x \in L$.
  Indeed, the YES player wins by setting $A_1 = L$; similarly, if $x \not\in L$, then the NO player wins by setting $A_0 = L$.
  Therefore, it is sufficient to show that we can compute the player that has a winning strategy in $\PSPACE$.

  We may restate the game as follows:
  Simulate $S$ on input $x$.
  If $S$ makes a query $x$ to $A_i \supplement{i \in \binset}$, then answer it with $i$.
  If $S$ makes a query $q \supplement{ \neq x }$ to the oracle $A_0$,
  then the NO player gives an arbitrary answer;
  similarly, if $S$ makes a query to $A_1$, then the YES player gives an arbitrary answer.
  (However, we require the players to behave in a consistent way:
  if $S$ makes the same query more than once, then a player must give the same answer that the player answered in the past.)

  Again, one can easily prove that the YES player has a winning strategy for this game if and only if $x \in L$.

  Now we describe a polynomial-time alternating Turing machine that computes $L$:
  Simulate the game described above,
  while universally guessing the answers of the NO player
  and existentially guessing the answers of the YES player.
  Since a polynomial-time alternating machine can be simulated in $\PSPACE$, it holds that $L \in \PSPACE$.
\end{proof}

\section{Random Strings vs.\ Randomized Computation}
\label{sec:app}
In this section,
we apply the notion of selector to the proof by Buhrman, Fortnow, Kouck\'{y}, and Loff~\cite{BFKL10}.
We thereby extend their result
from $\ESet{\NP,\allowbreak\PTIME^{\#\PTIME},\allowbreak\PSPACE,\allowbreak\EXP}$ to any classes whose complete languages have a selector
($\eg \allowbreak \PH i,\allowbreak \coPH i,\allowbreak\PTIME^{\#\PTIME},\allowbreak\PSPACE,\allowbreak\EXP,\allowbreak$ and $\EXPNP$).

\begin{theorem}[Extended Theorem 15 of \cite{BFKL10}]
  Let $\alpha \colon \{0\}^* \to \{0, 1\}^*$ be a length preserving function,
  $c > 0$ be a constant such that $\alpha(0^n) \not\in \io{\EXP}/n-c\log n$,
  and $\mathcal C$ be a complexity class such that there is a selector for some paddable $\mathcal C$-complete language $L$.
  If
  $L \in \PTIME/\alpha(0^{n^d})$ for some $d > 0$,
  then
  $\mathcal C \subset \BPP$.
\end{theorem}

\begin{proof}
  Let $M$ be a polynomial-time machine such that $L(x) = M \left(x, \alpha(0^{|x|^d}) \right)$,
  and
  $G_n \subset \binset^{n^d}$ be the set of ``good'' advice:
  \[
    G_n := \ISet{r \in \binset^{n^d}} { \forall x \in \binset^n, \, L(x) = M(x, r)  }.
  \]
  \etal{Buhrman}~\cite{BFKL10} showed that $|G_n| \ge 2^{n^d}/n^{cd}$ by exploiting the high nonuniform complexity of advice $\alpha(0^{n^d})$.

  As with Theorem \ref{thm:finercharact},
  there exist a polynomial $l$ and a selector $S$ that
  identifies an honest oracle among $m := 2 l(n) ^{cd}$ oracles with probability at least $\frac56$,
  and
  makes only queries of length exactly $l(n)$ on inputs of length $n$.

  Consider the following probabilistic algorithm:
  On input $x \in \binset^n$,
  let $l$ denote $l(n)$.
  We pick $m$ random strings $r_1, \cdots, r_m \in_R \binset^{l^d}$ uniformly at random,
  and define oracles $A_i(q) = M(q, r_i)$, for any $i \in \NumSet m$ and for any $q \in \binset^{l}$.
  We simulate $S$ on input $x$, answering its queries $q \in \binset^l$ to $A_i$ by computing $M(q, r_i)$.

  The probability that we fail to pick any ``good'' advice, namely $r_i \not\in G_l$ for all $i$,
  is $\left(1 - |G_l|\right)^{2l^{cd}} \le e^{-2l^{cd}/l^{cd}} < \frac16$.
  Thus, we can output the correct answer with probability at least $\frac23$ overall.
\end{proof}

\section{Concluding Remarks}
\label{sec:conclusions}
We state some open problems and possible directions for future work:
\begin{itemize}
  \item
    Do there exist selectors for $\NEXP$-complete languages or promise-$\SHEXP$-complete languages?
    In particular, it is interesting to close the gap between $\EXPNP$ and $\SHEXP$:
    although these classes seem ``close'' in some sense,
    $\EXPNP$ and $\SHEXP$ are very different in the known relationship with $\BPP$;
    it is a notorious open problem whether $\BPP \neq \EXPNP$, whereas one can prove $\BPP \neq \SHEXP$.
  \item
    We proved that a property of removing short advice can be captured by the notion of selector.
    What about a property of removing advice of polynomial length?
  \item
    The result of $\MIP = \NEXP$ was ``scaled-down'' to obtain the relationship with hardness of approximating cliques \cite{FGLSS96},
    and eventually the PCP theorem \cite{AS98,ALMSS98} was established.
    Can we obtain such interesting applications of selectors, by scaling down the selector for $\EXPNP$-complete languages?
\end{itemize}

\subparagraph*{Acknowledgements}
I greatly appreciate Hiroshi Imai's advice and comments that significantly improved the presentation;
I thank Akitoshi Kawamura for many useful discussions;
I am deeply grateful to Lance Fortnow and the anonymous CCC reviewers for very helpful comments 
that made the paper more understandable;
and I would like to thank the reviewer for suggesting the title.

\bibliography{bib}

\end{document}